\documentclass[submission,copyright,creativecommons]{eptcs}
\providecommand{\event}{AFL 2026} 

\usepackage{amsmath,amssymb}

\usepackage{amsthm}

\ifpdf
  \usepackage{underscore}         
  \usepackage[T1]{fontenc}        
\else
  \usepackage{breakurl}           
\fi
\newtheorem{definition}{Definition}
\newtheorem{theorem}{Theorem}

\newtheorem{corollary}{Corollary}

\newtheorem{proposition}{Proposition}

\newtheorem{example}{Example}

\newcommand {\5} {$5'\to 3'$ }
\newcommand {\3} {\mathcal{L}_{53SAS}}
\newcommand {\hash} {\#}

\title{Reverse Post Correspondence Problem and Undecidability of $5'\to 3'$ String Assembly Systems}
 \author{Benedek Nagy
 \institute{
 	Department of Mathematics, Faculty of Arts and Sciences, \\
      Eastern Mediterranean University, Famagusta, North Cyprus, via Mersin-10, T\"urkiye\\ and \\
      Department of Computer Science, Institute of Mathematics and Informatics\\ Eszterh\'azy K\'aroly Catholic University, Eger, Hungary}
 \email{nbenedek.inf@gmail.com}\\
 }

 \def\titlerunning{Reverse PCP and Undecidability of $5'\to 3'$ SAS}%
 \def\authorrunning{Benedek Nagy}

\begin{document}
\maketitle

\begin{abstract} The
Post Correspondence Problem is as follows: 
 having a set of dominoes, %
is there any (maybe repeating) sequence of them such that the words formed by the upper parts and the lower parts by the sequence of dominoes are identical. It is one of the most known problems that is algorithmically undecidable. In this paper,
the reverse Post Correspondence Problem is defined, that is, where the two assembled words of the dominos are reversals of each other. Undecidability about this new, modified problem is proven.

Further, based on this result, it is also proven that the emptiness %
 problem for \5 String Assembly Systems is undecidable too. \5 String Assembly Systems  belong to String Assembly Systems type formal language generating models. The \5 denotes that, in these variants, the derivations of the generated words start from the two extremes. The notation and the new model are bio-motivated as any double stranded DNA has two opposite oriented \5 strands. 
\end{abstract}


\section{Introduction} %

Formal languages and automata theory are basic fields of computability and computer science. Turing, in his seminal paper \cite{Turing}, on the one hand, defined a universal model of algorithm, later named after him, the Turing machine, and on the other hand, he proved that there are problems that are algorithmically not solvable, i.e., not computable, e.g., the halting problem of Turing machines.
 Later, among others, Post published an easy to understand problem, the Post correspondence problem (PCP) that is, similarly to the halting problem of Turing machines, undecidable, i.e., not computable \cite{Post}. In this problem a finite set of dominoes is given with words written both in their upper and lower parts. The question is whether one can make a sequence of dominoes (maybe using any domino more than once) such that the strings that are composed 
 in the upper and lower parts are identical.
 The PCP has many variants, including recent ones, some of them can algorithmically be solved \cite{Halava-IJFCS}, while some of them, similarly to the original variant, are algorithmically not solvable in general \cite{RAIRO23}. 
Among many variants, we recall only the circular PCP \cite{circPCP}, where, given the set of dominoes, the question is whether
 the string composed in the lower part %
  can be a cyclic permutation of the string composed in the upper part. 
 This problem is also undecidable.
In our paper, a new variant of PCP, the Reverse PCP, will play a central role. In a nutshell, Reverse PCP asks if given an instance of the PCP, is it possible to
arrange the dominoes so that one word is the reverse of the other (the formal definition comes later in the paper). 

 DNA computing emerged some decades ago including various branches of new computing paradigms \cite{DNAbook}.
Various models are based on the double-stranded and complementary structure of DNA molecules, including sticker systems, Watson-Crick automata \cite{Freund} and string assembly systems (SAS). It was proven in \cite{Kuske} that the Watson-Crick complementarity relation can be excluded from sticker systems and Watson-Crick automata without affecting the defined languages, therefore SAS is already defined without involving this relation into it  similarly as \5 Watson-Crick automata are used, e.g., in \cite{CiE2009,2detLIN,AFL2017}. 
SAS are defined in \cite{SAS}, and some of their variants are further developed in \cite{SAS2,SAS3}. Their decidability is addressed, e.g., in \cite{SASdec}.
In this paper, we present a variant that is motivated by the structure of the DNA strands, namely by the fact that they have opposite \5 directions. Consequently, when an enzyme is building and constructing a strand, it uses the same biochemical direction for both strands; which directions, in fact, are opposite physical (or mathematical) directions. 
In this paper, we investigate %
 this new variant. Actually, here we focus on decision problems;  in another paper \cite{53SAS-hier}, relation to other classes of languages and hierarchy results are discussed.

Up to now, due to our best knowledge, there was no variant of PCP where the two strings are reversal of each other. The proof of undecidability of these problems and some decision problems of \5 SAS are closely related as we will use very similar constructive proofs for them.
  
In this paper, after we recall some formal concepts and we give our notations in Section \ref{sec:pre}, we define the reverse PCP  and prove its undecidability (Section \ref{sec:rPCP}). In Section \ref{sec:53SAS}, we use a similar proof %
 to show that the emptiness  problem for \5 String Assembly Systems is also undecidable.
We also show that %
in the unary case 
these %
problems
are decidable.

\section{Preliminaries}\label{sec:pre}

It is expected that the reader is familiar with the basics of computability including formal languages and
automata. Otherwise, for unexplained basic concepts the reader is referred, e.g.,
to \cite{HopUl,%
Sipser,DNAbook}. Here, we recall some concepts that play central roles in this paper and also
fix our notations.

Let $\mathbb N$ denote the set of positive integers and let $\lambda$ denote the empty word.
For an alphabet $\Sigma$, $\Sigma^*$ contains all words (also called strings) that can be obtained by its symbols
including the empty word and $\Sigma^+$ denotes the set of all nonempty words that can be obtained by $\Sigma$.
In this paper, we generally work with $\lambda$-free languages, and therefore, two languages are considered to be equal if they differ at most in the empty word $\lambda$.
The length of a word
$w\in \Sigma^*$  is denoted by $|w|$. The reversal of a word $w=a_1\dots a_n$ is $w^R=a_n\dots a_1$ with $a_i\in \Sigma$.

A Turing machine (abbreviated as TM) is a finite state computational device with a possibly infinite tape and a reading-writing head. The head is always over a tape cell. In any instant of the time, the tape has only a finite segment with valid content, all the remaining infinite parts both to the left and to right are containing special space  symbols (denoted by $\hash$).
Formally, a TM is an ordered tuple $$(Q,V,\Sigma,q_0,\hash,\delta,Q_f),$$ where
\begin{itemize}
\item $Q$ is a finite set of states;
\item $V$ is the input alphabet;
\item $\Sigma$ is the tape alphabet, $V \subset \Sigma$, w.l.o.g. we assume that $\Sigma \cap Q =\emptyset$;
\item $q_0 \in Q$ is the initial state;
\item $\hash \in \Sigma\setminus V$ is the space symbol (denoting the content of tape cells that are not in use);
\item $\delta: Q\times \Sigma \to Q \times \Sigma \times \{left,stay,right\}$ is the transition function (the program of the TM); and finally
\item $Q_f$ is the set of final (also called halting) states; a subset $Q_a\subset Q_f$ %
 is used as accepting states if the TM is used to accept a language (then $Q_a$ may appear in the tuple directly instead of $Q_f$).
\end{itemize}

 The computation is usually defined by configurations. The initial configuration on an input $w\in V^+$ is represented as $q_0 w$.
A general configuration is a string $uqv\in \Sigma^* \cdot Q \cdot \Sigma^+$ (with $q\in Q$), where it is assumed that the head is over the first symbol of $v$. It is assumed that $u$ never starts with $\hash$ (i.e., any prefix of the form $\hash^+$ is not written). On the other hand, any suffix of the form $\hash \hash^+$ is not written for $v$, but it is allowed to have $v=\hash$, in case the head is %
on the right end of the valid tape content. %

The computation steps are defined by the transition function as follows:
\begin{itemize}
\item $uqav \Rightarrow uq'bv $, if $(q',b,stay) = \delta(q,a)$ with $a,b\in \Sigma$, $q\in Q\setminus Q_f$, $q'\in Q$, $u,v\in\Sigma^*$;
\item $uqav \Rightarrow ubq'v $, if $(q',b,right) = \delta(q,a)$ with $a,b\in \Sigma$, $q\in Q\setminus Q_f$, $q'\in Q$, $u,v\in\Sigma^*$, $v\ne \lambda$;
\item $uqa \Rightarrow ubq'\hash $, if $(q',b,right) = \delta(q,a)$ with $a,b\in \Sigma$, $q\in Q\setminus Q_f$, $q'\in Q$, $u \in\Sigma^*$;
\item $ucqav \Rightarrow uq'cbv $, if $(q',b,left) = \delta(q,a)$ with $a,b,c\in \Sigma$, $bv\not\in \hash^*$, $q\in Q\setminus Q_f$, $q'\in Q$ $u,v\in\Sigma^*$;
\item $ucqa \Rightarrow uq'c $, if $(q',\hash,left) = \delta(q,a)$ with $a,c\in \Sigma$, $q\in Q\setminus Q_f$, $q'\in Q$, $u\in\Sigma^*$;
\item $qav \Rightarrow q' \hash bv $, if $(q',b,left) = \delta(q,a)$ with $a,b\in \Sigma$, $bv\not\in \hash^*$, $q\in Q\setminus Q_f$, $q'\in Q$ $v\in\Sigma^*$;
\item $qa \Rightarrow q'\hash $, if $(q',\hash ,left) = \delta(q,a)$ with $a\in \Sigma$, $q\in Q\setminus Q_f$, $q'\in Q$.
\end{itemize}
The computation relation $\Rightarrow^*$  is defined as the reflexive and transitive closure of the computation step.
Without loss of generality, it is assumed that the tape contains always at most one consecutive segment containing a word of $(\Sigma\setminus \{\#\})^+$.
We say that the computation of the TM halts on an input $w$ if starting from its initial configuration $q_0 w \Rightarrow^* u p v$ with $p\in Q_f$. Without loss of generality, in this paper, we assume that in halting configurations the head is at the end of the tape content, i.e., $v=\#$.

The halting problem of Turing machines is the first algorithmically undecidable problem: as it has been proven by Turing in his seminal paper \cite{Turing}, it is generally, algorithmically undecidable whether a TM is halting on a given input, more precisely:

\begin{theorem}{\bf (Turing)}\label{thm:T} Let %
a Turing machine, and an input word for this Turing machine be given.
There is no algorithm that can decide for arbitrary pairs of Turing machine and input word if the given Turing machine halts on this %
 input.
Further, there is no algorithm that is able to decide if a given (as input) Turing machine  has any word to accept, i.e., whether the language accepted by the Turing machine is non-empty.
\end{theorem}

A decade later, Post published another algorithmically undecidable problem, which is also one of the most known such problems \cite{Post}.

In a nutshell, the Post Correspondence Problem (PCP) is as follows:

Let a finite set $D=\{d_1,\dots,d_n\}$ of dominoes, i.e., string pairs be given. The problem is to decide whether there exists a (maybe repeating) sequence of dominoes such that the strings built in the upper and lower parts coincide. Formally,
let the dominoes, the
pairs of words $ D =\{ (\alpha_1, \beta_1), (\alpha_2, \beta_2),\dots, (\alpha_n, \beta_n) \}$ be given %
where $d_i$ has upper part $\alpha_i$ and lower part $\beta_i$, respectively. %
Is there any nonempty finite
sequence $i_1\dots i_k$ of the dominoes, i.e., of indices %
(in fact,  an indexword over the alphabet $\{1,\dots, n\}$) such that $\alpha_{i_1}\alpha_{i_2} \dots \alpha_{i_k} = \beta_{i_1}\beta_{i_2} \dots \beta_{i_k}$. In the affirmative case, if there is a solution, the given indexword (order of dominoes) defines also the word $w=\alpha_{i_1}\alpha_{i_2} \dots \alpha_{i_k}$  ``generated'' by the system.

\begin{theorem}{\bf (Post)} Let the input be a finite set of dominoes (string pairs).
There is no algorithm that can decide for arbitrary input whether this PCP has a solution.
\end{theorem}

Further, let us recall the original concept of SAS from \cite{SAS}.

\begin{definition}\label{def:SAS}
 A \emph{string assembling system} (SAS, shortly) is a quadruple $(\Sigma, A, T, E)$,
where \begin{enumerate}
\item $\Sigma$ is the alphabet, that is, a finite, nonempty set of symbols or letters;
\item $A \subset \Sigma^+ \times \Sigma^+$  is the finite set of axioms of the forms $(uv, u)$ or $(u, uv)$, where
$u\in \Sigma^+$ and $v \in \Sigma^*$;
\item $T \subset \Sigma^+ \times \Sigma^+$ is the finite set of assembly units; and finally,
\item $E \subset \Sigma^+ \times \Sigma^+$  is the finite set of ending assembly units of the forms $(vu, u)$ or
$(u, vu)$, where $u\in \Sigma^+$ and $v \in \Sigma^*$.
\end{enumerate}

Then, formally the units are assembled as follows. 
Let $S = (\Sigma, A, T, E)$ be a SAS.

The derivation relation $\Rightarrow$  is defined on specific subsets of  $\Sigma^+ \times \Sigma^+$  by
\begin{enumerate}
\item $(uv, u) \Rightarrow (uvx, uy)$ if
\begin{enumerate}
\item $uv = ta$, $u = sb$, and $(ax, by) \in T \cup  E$, for $a, b \in \Sigma$, $x, y, s, t \in \Sigma^*$; and
\item $vx = yz$ or $vxz = y$, for $z \in \Sigma^*$.
\end{enumerate}
\item
$ (u, uv) \Rightarrow (uy, uvx)$ if
\begin{enumerate}
\item $uv = ta$, $u = sb$, and $(by, ax) \in T \cup E$, for $a, b \in \Sigma$, $x, y, s, t \in\Sigma^*$; and
\item $vx = yz$ or $vxz = y$, for $z \in \Sigma^*$.
\end{enumerate}
\end{enumerate}

A derivation is said to be successful if it initially starts with an axiom from $A$,
continues with assembling units from $T$, and ends with assembling an ending unit
from $E$. The process necessarily stops when an ending assembly unit is added. 
When an element of $E\cap T$ is used, the process may stop or may continue (nondeterministically). The
sets $A$, $T$, and $E$ are not necessarily disjoint.

The language $L(S)$ generated by $S$ is defined to be the set
$$L(S) = \{ w \in\Sigma^+  | (p, q) \Rightarrow^* (w, w) \text{ is a successful derivation}\},$$
where, as usual, $\Rightarrow^*$ refers to the reflexive and transitive closure of the derivation relation $\Rightarrow$.
\end{definition}

\begin{theorem}{\bf (Kutrib and Wendlandt)} Let the input be a SAS.
There is no algorithm that can decide for arbitrary input SAS if there is any string that can be generated by this SAS.
\end{theorem}

We recall that there are variants of SAS called  free SAS, one-set SAS and pure SAS \cite{SAS2}, such that based on the direct connection of these systems to PCP, it is undecidable for each of them whether the generated language is nonempty.

 Since in this paper  we focus on \5 SAS, we define them below.

\subsection{Definitions of \5 SAS} %

In this paper, we investigate a new variant, the $5'\to 3'$ SAS, in which, from mathematical %
point of view, in the beginning, the two strands are growing somewhat independently in opposite directions. However, when a nonempty suffix of the first strand and a nonempty prefix of the second strand coincide, the two strands may join on this part. After the strands are joined, they must match to the already generated other strand: and finally, the generation produces a string of the defined language if and only if %
 the
two strings of the strands coincide.

Then, formally we define our new concept as follows.
\begin{definition}\label{def:53SAS}
A $5'\to 3'$
\emph{string assembling system} ($5'\to 3'$ SAS, shortly) is a quadruple $(\Sigma, A, T, E)$,
where %
\begin{itemize}
\item $\Sigma$ is the alphabet, that is, a finite, nonempty set of symbols or letters;
\item $A, T, E  \subset \Sigma^+ \times \Sigma^+$  are the finite,
 not necessarily disjoint %
sets of axioms, %
assembly units, and %
ending assembly units, respectively. %
\end{itemize}

Let $S = (\Sigma, A, T, E)$ be a $5'\to 3'$ SAS.
The derivation relation $\Rightarrow$  is defined on %
$\Sigma^+ \times \Sigma^+$  by
 $(u, v) \Rightarrow (ux,yv)$ %
 if %
 $u = ta$, $v = bs$, and $(ax, yb) \in T \cup  E$, for $a, b \in \Sigma$, $x, y, s, t \in \Sigma^*$. (Thus, in fact, $(ux,yv) = (tax, ybs)$.)

A derivation is said to be successful if it initially starts with an axiom, an element of $A$,
continues with assembling units from $T$ (if any), and ends with an ending assembling unit
from $E$. The process necessarily stops when an ending assembly unit from $E\setminus T$ is added.
When an element of $E\cap T$ is used, the process may stop or may continue.
 The sets $A$, $T$, and $E$ are not necessarily disjoint.

The language $L(S)$ generated by $S$ is defined to be the set
$$L(S) = \{ w \in\Sigma^+  | (p, q) \Rightarrow^* (w, w) \text{ is a successful derivation}\}.$$
\end{definition}

The class of languages that can be generated by $5'\to 3'$ SAS is called $5'\to 3'$ SAS languages and denoted by $\mathcal{L}_{53SAS}$.

If the alphabet has only one letter, then the corresponding $5'\to 3'$ SAS is a \emph{unary $5'\to 3'$ SAS}.  

Further in this section, to show the generative power of our newly defined systems, we give some characteristic examples.

\begin{example}\label{exa-nnn}
Let $S= (\{1,2,3\},\{(1,3)\},\{(11,33),(12,23),(22,2),(23,2),(33,22), (3,12),(3,11)\},$ \\ $\{(3,1)\})$.
The  language $L(S)$ generated by $S$ is $\{1^n 2^n 3^n ~|~ n \in \mathbb N \}$. Each derivation starts with the sole axiom fixing that any word of the language must starts with a $1$ and finishes with a $3$. In the first phase of the derivation, by $(11,33)$, the same number, let us say $n$, $1$-s and $3$-s are generated as the prefix and the suffix of the word. Then, applying the derivation step by $(12,23)$, in the second phase using $(22,2)$ let us say $m$ times (with $m\geq 0$), the block of $2$-s is built in the upper strand, and this will be between the blocks of $1$s and $3$s. By using $(23,2)$, the derivation enters in its third phase, where with $(33,22)$, the length of the block of $3$s is checked (by building it in the upper strand) to the block of $2$s which is now built in the lower strand. To have a successful derivation, $m=n$ must hold, since after applying unit $(3,12)$, the upper strand cannot be further extended. In this last phase of the derivation, the lower strand is extended from right to left to build the block of $1$s, which is the prefix of the derived word. As the derivation is successful only when the upper and lower strands are identical, each word of the language has the three blocks in the given order and these blocks must have equal lengths.
The language $L(S)$ is a well-known  context-sensitive and not context-free language.\qquad \qquad \qquad \qquad  \qquad \qquad \qquad \qquad \qquad \qquad \qquad \qquad \ \ \ \ \ $\circ$
\end{example}

\begin{example}\label{exa-pali}
Let $S= (\{1,2\},\{(1,1),(2,2)\},\{(111,111),(112,211),(121,121), (211,112),(221,122),$ \\ $(212,212),(122,221),(222,222)\},\{(1,1),(2,2)\})$. The language $L(S) = \{w \in \{1,2\}^* ~|~ |w| \text{ is odd and }$ $ w= w^R \}$,  where $w^R$ is the reversal of the word $w$.
By the axioms, any word of the language must have the same letter as initial and last letter. Then, step by step, as longer prefix and suffix are determined, they are always the reversal of each other, moreover, the length (of each strand) is always increasing by $2$.
This language is the language of odd-length palindromes over $\{1,2\}$, 
 a well-known linear
language in 2detLIN \cite{2detLIN}. \qquad
  \qquad \qquad \qquad \qquad \qquad \qquad \qquad \qquad \qquad \qquad \qquad \qquad \qquad \qquad \qquad \qquad \quad \  $\circ$
\end{example}
In a similar manner even palindromes and also the set of palindromes can be generated. These examples also highlight the difference between SAS and \5 SAS, as it is shown in \cite{SAS} that (some variants of) the language of palindromes cannot be generated by any SAS. For further details about relations of $\3$ and other classes of languages, we refer to \cite{53SAS-hier}.
  
On the one hand, as we have shown in our examples \5 SAS can generate some relatively complex languages, however, on the other hand, there are regular languages that cannot be generated by them, e.g., the language $1+11+(111)^+$ is not in $\3$: %

\begin{proposition}\label{prop0:notIn}
The unary regular language $L$ described by the regular expression $1+11+(111)^+$ cannot be generated by any \5 SAS.
\end{proposition}

\begin{proof}
The proof is by contradiction, thus let us assume that there is a \5 SAS $S$ that generates $L$.

Since $1\in L$, both $A$ and $E$ must have the pair $(1,1)$.
Now, by the finite sets $A$ and $E$ without using any elements of $T$, only finitely many elements of $L$ can be obtained, thus we need some elements in $T$.
Then we have the following cases:
\begin{enumerate}
\item Let us assume that there is an element $(1^n,1^n) \in T$ for some $n > 1$. (The case $n=1$ is useless for generation.) %
Then, applying it once together with the axiom and ending unit $(1,1)$, the word $1^n$ is derived. Thus $1^n \in L$ must hold. Then there are two subcases.
\begin{enumerate}
\item If $n=2$, then, in fact, with this assembling unit any element of $1^+$ can be obtained, i.e., applying it $k \geq 0$ times, the word $1^{k-1}$ is derived. However, this contradicts to the fact that there are some words, e.g., $1^4$ and $1^5$ (%
of $1^+$), that are not in $L$.
\item If $n = 3k$ for some positive integer $k$, then clearly $1^n \in L$. Now, applying this unit twice in the derivation (with the same axiom and ending assembly unit $(1,1)$) the word $1^{2n-1}$ is derived. However, this word is not in $L$. We have arrived again to a contradiciton. Thus, $T$ cannot contain any unit with the same length, i.e., identical words (apart from the useless $(1,1)$).
\end{enumerate}
\item If there is an element $(1^n,1^m) \in T$ with the property $n>m$, then there must also be an element $(1^k,1^j)$ with $k<j$, otherwise only finitely many derivations exist contradicting to the infinity of $L$. Now starting and ending with unit $(1,1)$ and applying unit $(1^n,1^m)$ $j-k$ times and the unit $(1^k,1^j)$ $n-m$ times, the word $1^{n(j-k)+k(n-m) - (j-k+n-m-1)} = 1^{j(n-m)+m(j-k)-(j-k+n-m-1)}$ can be derived, thus it must be in the language. Thus, $nj-mk-(j-k)-(n-m)+1$ is either 2 or a multiplier of 3. However, here we can use the same argument as for item 1 above, i.e., by using double, triple, or 4-times more the elements, such words are also generated that do not belong to $L$. That is a contradiction.
\item Finally, if there is a unit $(1^n,1^m) \in T$ with $n<m$, then, similarly to the previous case, either only finitely many derivations (and derived words) exist or there must also be an element $(1^j,1^k) \in T$ with $j>k$. In the latter case the same argument works as in case 2.
\end{enumerate} 

The proof has been finished. %
\end{proof}

\section{The reverse PCP}\label{sec:rPCP}

This section has two parts, first we formally define the reverse PCP and then we shall prove an undecidability result about it.

Let a finite set of dominoes $D= \{d_1,\dots,d_n\}$ be given, where each domino is represented by a pair $d_i = (\alpha_i,\beta_i) \in \Sigma^* \times \Sigma^*$. For a sequence of dominoes  $w \in D^*$ (or, similarly, to an indexword)
the upper word $u(w)$ is assigned as follows:
 \begin{itemize} 
\item
if $w=\lambda$, then so $u(w)=\lambda$; 
\item
 if $w=w_0 d_i$ for $w_0\in D^*$, %
then $u(w)=u(w_0) \cdot \alpha_i$.
\end{itemize}
In a similar manner the lower word $\ell(w)$ is also assigned as %
follows: 
\begin{itemize}
\item
if $w=\lambda$, then so $\ell(w)=\lambda$;
\item
if $w=w_0 d_i$ for $w_0\in D^*$, %
then $\ell(w)=\ell(w_0) \cdot \beta_i$.
\end{itemize}

The reverse Post Correspondence Problem (reverse PCP) is to decide whether there exists a finite sequence $w \in D^+$ of dominoes such that the assigned upper and lower words are reversals of each other, i.e., $u(w) = (\ell(w))^R$.

Let us define another variant of the PCP that is closely connected to the reverse PCP.

The opposite direction Post Correspondence Problem (odPCP) is as follows.

Given the finite set of dominoes (as above), their upper parts are placed next to each other from left to right (as usual, and thus the upper word is defined exactly as for the classical and also for the reverse PCP) and the lower parts are placed next to each other from right to left, formally:
$\ell'(w)$ is assigned to $w\in D^*$ as
follows: \begin{itemize}
\item if $w=\lambda$, then so $\ell'(w)=\lambda$;
\item if $w=w_0 d_i$ for $w_0\in D^*$,
then $\ell'(w)=\beta_i \cdot \ell'(w_0)$.
\end{itemize}
The decision problem of odPCP is to decide whether there exists a finite sequence $w \in D^+$ of dominoes such that the assigned upper and lower words are the same, i.e., $u(w) = \ell'(w)$.

\begin{theorem}\label{thm: rev-od}
The reverse PCP and the odPCP are equivalent to each other, that is
\begin{itemize}
\item for any reverse PCP instance $\Gamma$, there is an odPCP instance $\Gamma'$ such that $\Gamma$ has a solution if and only if $\Gamma'$ has a solution, moreover there is a bijection between their solution sets; and
\item for any odPCP instance $\Gamma$, there is a reverse PCP instance $\Gamma'$ such that $\Gamma$ has a solution if and only if $\Gamma'$ has a solution, moreover there is a bijection between their solution sets.
\end{itemize}
\end{theorem}

\begin{proof} Instead of a formal proof, we give a hint: %
We show a bijection between the instances of these problems.
If, in each domino, instead of the original lower word $\beta_i$, its reversal $\beta_i^R$ is used, then using these modified dominoes,
the problem is switched from odPCP to reverse PCP and vice versa. %
\end{proof}

Therefore, the reverse Post Correspondence Problem is equivalent to a PCP in which the directions to put the upper and lower parts of the dominoes next to each other is opposite to each other. That is a good point for our aim, i.e., to connect a PCP type decision problem to \5 SAS. We will come back to this point in the next section, however, first let us show undecidability results for these new PCP variants.

\subsection{Undecidability of odPCP}\label{sec:subUndecPCP} %

Let a TM be given as the tuple $(Q,V,\Sigma,q_0,\hash,\delta,Q_f)$.

The idea of the proof is to embed Turing machine computations into the upper and lower words of our PCP. For this we will write the configurations of the Turing machine by strings of the form $uqv$ where $u\in\Sigma^*$, $q\in Q$ and $v\in \Sigma^+$ such that $u$ does not contain any (prefix) $\#^+$ and $v\in (\Sigma\setminus \{\#\})^+ \cup \{\#\}$ as we have described in the preliminaries.  

Consequently, the alphabet for the reverse PCP is $\Sigma_1 = \Sigma \cup Q$. Further, for technical reasons, let us make disjoint copies of $\Sigma_1$: $\Sigma_2 = \{a'~|~a\in \Sigma_1\}$ with $\Sigma_1 \cap \Sigma_2 = \emptyset$ and
$\Sigma_3 = \{a''~|~a\in \Sigma_1\}$ with $\Sigma_1 \cap \Sigma_3 = \emptyset$ and $\Sigma_2 \cap \Sigma_3=  \emptyset$.
 Further, let $\Sigma_4 = \{a'''~|~a\in \Sigma_1\}$ with $\Sigma_1 \cap \Sigma_4 = \emptyset$, $\Sigma_2\cap \Sigma_4=  \emptyset$ and $\Sigma_3\cap \Sigma_4=  \emptyset$.
Now let $\Sigma' = \Sigma_1 \cup \Sigma_2 \cup \Sigma_3 \cup \Sigma_4
\cup \{ \$, @ \}$ with separator symbols $\$ , @ \not\in \Sigma_1 \cup \Sigma_2 \cup \Sigma_3 \cup \Sigma_4$ the alphabet of the odPCP.

For a given input $w\in V^+$, the computation starts with configuration $c_0 = q_0 w$.
Now, our aim is to prepare a set of dominos for an odPCP that can be assembled in a way such that they describe halting computations of a given TM in the following way in both strands: $$\$c_0 \$ c'_1 \$ \dots \$ c'_{n-1} \$ c'''_n \$ \$ @ (c''_{n})^R@ (c''_{n-1})^R  @ \dots @ (c''_1)^R,$$ where the halting computation is $c_0 \Rightarrow c_1 \Rightarrow \dots \Rightarrow c_{n}$ and the primed, double and triple primed versions of the configurations are written with the alphabets $\Sigma_2$, $\Sigma_3$ and $\Sigma_4$, respectively.

Thus, let us design the following set $D$ of dominoes:

To build $c_0$ and the first computation step by making also $(c''_1)^R$ at the same time:
\begin{itemize}
\item $(\$q_0a,b''q'')$ if $(q,b,stay) = \delta(q_0,a)$ with $a\in V$, $b\in \Sigma$, $q_0\in Q\setminus Q_f$, $q\in Q$, where $b'',q''\in \Sigma_3$ are representing $b$ and $q$;
\item $(\$q_0a,q''b'')$ if $(q,b,right) = \delta(q_0,a)$ with $a\in V$, $b\in \Sigma$, $q_0\in Q\setminus Q_f$, $q\in Q$, $b'',q''\in \Sigma_3$ are representing $b$ and $q$;
\item $(\$q_0a,b''\#'' q'')$ if $(q,b,left) = \delta(q_0,a)$ with $a\in V$, $b\in \Sigma$, $q_0\in Q\setminus Q_f$, $q\in Q$, $b'',q'',\#'' \in \Sigma_3$ are representing $b$, $q$ and $\#$;
\end{itemize}
and to copy the remaining input to the new configuration:
\begin{itemize}
\item $(a,a'')$ for all $a\in V$ where $a''$ represents $a$ in $\Sigma_3$.
\end{itemize}
Then putting the markers $(\$,@)$.

Now, in a similar manner, we build dominoes for assuring that the configuration built in the upper strand (in a primed manner) is followed in the computation by the configuration built in the lower strand (in reversed and double primed manner).
\begin{itemize}
\item $(a',a'')$ for each $a\in\Sigma$ to write and copy the part of the tape that is before and after the head;
\end{itemize}
And for the part where the computation step changes the configuration:
\begin{itemize}
\item $(q'a', b''p'') $, if $(p,b,stay) = \delta(q,a)$ with $a,b\in \Sigma$, $q\in Q\setminus Q_f$, $p\in Q$ with their primed and double primed versions in $\Sigma_2$ and $\Sigma_3$, respectively;
\item $(q'a'd' , d''p''b'') $, for all $d\in \Sigma$ if $(p,b,right) = \delta(q,a)$ with $a,b\in \Sigma$, $q\in Q\setminus Q_f$, $p\in Q$  with their primed and double primed versions in $\Sigma_2$ and $\Sigma_3$, respectively;
\item $(q'a'\$ , @\#''p''b'') $, if $(p,b,right) = \delta(q,a)$ with $a,b\in \Sigma$, $q\in Q\setminus Q_f$, $p\in Q$  with their primed and double primed versions in $\Sigma_2$ and $\Sigma_3$, respectively (for the case when there was a right move from the last used tape cell);
\item $(c'q'a' ,b''c''p'') $, for all $c\in\Sigma$ if $(p,b,left) = \delta(q,a)$ with $a,b\in \Sigma$, $b\ne \hash$, $q\in Q\setminus Q_f$, $p\in Q$ with their primed and double primed versions in $\Sigma_2$ and $\Sigma_3$, respectively;
\item $(c'q'a' ,c''p'') $, for all $c\in\Sigma$ if $(p,\#,left) = \delta(q,a)$ with $a\in \Sigma$,  $q\in Q\setminus Q_f$, $p\in Q$ with their primed and double primed versions in $\Sigma_2$ and $\Sigma_3$, respectively;
\item $(q'a' ,b''\#''p'') $,  if $(p,b,left) = \delta(q,a)$ with $a,b\in \Sigma$, $b\ne \hash$, $q\in Q\setminus Q_f$, $p\in Q$ with their primed and double primed versions in $\Sigma_2$ and $\Sigma_3$, respectively;
\item $(q'a' ,\#''p'') $,  if $(p,\#,left) = \delta(q,a)$ with $a\in \Sigma$,  $q\in Q\setminus Q_f$, $p\in Q$ with their primed and double primed versions in $\Sigma_2$ and $\Sigma_3$, respectively.
\end{itemize}

In this way, from the dominoes in the upper strand a sequence $\$ c_0\$ c'_1\$ \dots \$ c'_{m-1} $ can be assembled while in the lower strand the string $ (c''_m)^R @ \dots @(c''_1)^R$. 
\\
So far, no common characters are used, thus the two strands cannot overlap up to this point.

Further, we  continue it to the halting configuration:

We assumed that at a halting computation the head of the TM is set to the right end of the input, this can always be done by 
a step to the right arriving over a $\#$. %
\\
Thus, we assume that TM reaches (one of its) final state(s) in a right step. To simulate such computation step (and the halting), the following dominoes are designed:
\begin{itemize}
\item $(q'a'\$ ,\$ \$ @\#''p''b'') $, if $(p,b,right) = \delta(q,a)$ with $a,b\in \Sigma$, $q\in Q\setminus Q_f$, $p\in Q_f$  with their primed and double primed versions in $\Sigma_2$ and $\Sigma_3$, respectively (%
when a halting configuration is reached). %
 \end{itemize}

With the dominoes
\begin{itemize}
\item $(a''',\lambda)$ for all $a\in \Sigma$ and %
\item $(p''' \#''' \$ \$ @, \lambda)$  for each $p\in Q_f$
\end{itemize}
a copy of the accepting configuration $c'''_n$ can be obtained on the upper strand.
This could be the first step when the two strands overlap.

So far, it was checked that the configuration written in $(c''_1)^R$ follows configuration $c_0$; moreover, each of the configurations of $(c''_{i+1})^R$ follows the configuration written in the upper strand as $c'_i$. However, none of the sequences are checked if they really follow each other or even, if they contain real configurations. We can be sure only about $c_0$ that it is a correct configuration (as it is written by $\Sigma_1$), and based on that also $(c''_1)^R$ must represent a valid configuration.

Thus, by building the missing parts of the two strands we need to check if they contain the same configurations, i.e., if $c'_2$ represents the same as $(c''_2)^R$ and, generally, if the configurations of $c'_i$ and $(c''_i)^R$ coincide for each $i\in\{2,\dots,n\}$.
\begin{itemize}
\item $(a'',a''')$ for each $a\in \Sigma_1$ with its primed and triple primed version;
\item $(a''@,\$a''')$ for each $a\in \Sigma$ with its primed and triple primed version;
\item $(a'',a')$ for each $a\in \Sigma_1$ with its primed and double primed version; and
\item $(a''@,\$a')$ for each $a\in \Sigma$ with its primed and double primed version.
\end{itemize}

Observe that in the lower strand elements of $\Sigma_2$ written exactly when corresponding elements of $\Sigma_3$ are written on the upper strand. Moreover,
to allow to have the markers between the configurations in a correct way, we allow them to put only at the end of a configuration (that is a nonempty string written by the corresponding alphabet).

In this way, when these parts of both strands are built, we can be sure that the configurations coded in the first half (left to right till the triple marker $\$\$@$) and the configurations from right to left are matching, i.e., they and their sequences $c_1,\dots,c_n$ are identical, respectively. %

Only one thing is missing,
to finish the generation of both strands we need to make the original initial configuration on the lower strand. For this, we make the dominoes:
\begin{itemize}
\item $(\lambda,a \$ )$ for each $a\in V$;
\item $(\lambda,a)$ for each $a\in V$; and finally,
\item $(\lambda, \$q_0a)$ for each $a\in V$ and with the initial state $q_0$.
\end{itemize}

Now, one can see that if we restrict the computations for obtaining double, identical strings only of the following regular form:
$$ \$q_0 V^+ \$ (\Sigma^+_2 \$)^* \Sigma^+_4 \$ \$ (@ \Sigma^+_3 )^+ ,$$
then, in fact, the odPCP `simulates' the accepting computations of the given Turing machine, as it is described above.

Based on the construction above and the well-known undecidability result, %
Theorem \ref{thm:T}, %
we have just proven the following theorem.

\begin{theorem}\label{thm:rPCP}
The opposite direction Post Correspondence Problem with regular filter is undecidable, i.e., there is no algorithm in general that can decide if for any given opposite direction PCP $\Gamma$ and a regular language $L$, $\Gamma$ has a solution that belongs to $L$.
\end{theorem}

Based on the equivalence of reverse PCP and odPCP (Theorem \ref{thm: rev-od}), we can also conclude:
\begin{corollary}
The reverse Post Correspondence Problem  with regular filter is undecidable, i.e., there is no algorithm in general that can decide 
 for any input pair of a reverse PCP $\Gamma$ and regular language $L$
whether  
  $\Gamma$ has a solution that belongs to $L$.
\end{corollary}

\section{Undecidability of \5 String Assembly Systems }\label{sec:53SAS}

In this section, we modify our previous construction for \5 SAS in a way that the regular filter will be `built in' the system.

\begin{theorem}  
The emptiness problem for the \5 SAS is generally %
undecidable.
\end{theorem}
\begin{proof}
Let a TM be given as $(Q,V,\Sigma,q_0,\hash,\delta,Q_f)$. Based on that, we construct a \5 SAS $S=(\Sigma',A,T,E)$ as follows.

The alphabet $\Sigma'$ is defined already %
 in Subsection \ref{sec:subUndecPCP}: with $\Sigma_1 = \Sigma \cup Q$; 
$\Sigma_2 = \{a'~|~a\in \Sigma_1\}$, %
$\Sigma_3 = \{a''~|~a\in \Sigma_1\}$, %
and
 $\Sigma_4 = \{a'''~|~a\in \Sigma_1\}$: %
 $\Sigma' = \Sigma_1 \cup \Sigma_2 \cup \Sigma_3 \cup \Sigma_4
\cup \{ \$, @ \}$ with separator symbols $\$ , @$. %

The starting assembly units, the axioms are defined as follows:

\begin{itemize}
\item $(\$q_0a,b''q'')$ if $(q,b,stay) = \delta(q_0,a)$ with $a\in V$, $b\in \Sigma$, $q_0\in Q\setminus Q_f$, $q\in Q$, where $b'',q''\in \Sigma_3$ are representing $b$ and $q$;
\item $(\$q_0a,q''b'')$ if $(q,b,right) = \delta(q_0,a)$ with $a\in V$, $b\in \Sigma$, $q_0\in Q\setminus Q_f$, $q\in Q$, $b'',q''\in \Sigma_3$ are representing $b$ and $q$;
\item $(\$q_0a,b''\#'' q'')$ if $(q,b,left) = \delta(q_0,a)$ with $a\in V$, $b\in \Sigma$, $q_0\in Q\setminus Q_f$, $q\in Q$, $b'',q'',\#'' \in \Sigma_3$ are representing $b$, $q$ and $\#$.
\end{itemize}

Then, $T$ has the following pairs to simulate the computations and build the upper and lower strings. %

\begin{itemize}
\item $(ab,b''a'')$ for all $a,b\in V$ where $a''$ and $b''$ in $\Sigma_3$ represent the letters corresponding to $a$ and $b$,
\item $(ab,b''q'')$ for all $a,b\in V$, $q\in Q$ where $b''$ and $q''$ in $\Sigma_3$ represent the letter corresponding to $b$ and the state corresponding to $q$.
\item $(a\$,@b'')$ for all $a \in V$ and $b''\in \Sigma_3$ to put the markers, after the first pair of configurations assembled in the upper and lower strand, respectively.
\end{itemize}

Then, for assembling strings representing the other pairs of configurations  (where the primed and double primed versions are used as %
in the previous proof):

\begin{itemize}
\item $(\$ a',a'' @)$ for each $a\in\Sigma$ (using their representations in $\Sigma_2$ and $\Sigma_3$, respectively); and
\item $(a'b',b''a'')$ for each $a,b\in\Sigma$ (using their representations in $\Sigma_2$ and $\Sigma_3$, respectively) to write and copy the part of the tape that is before and after the head;
\end{itemize} %
For the part where %
 the configuration is changed:

\begin{itemize}
\item $(\$ q'a', b''p'' @) $, if $(p,b,stay) = \delta(q,a)$ with $a,b\in \Sigma$, $q\in Q\setminus Q_f$, $p\in Q$; %
\item $(\$ q'a'd' , d''p''b'' @) $, for all $d\in \Sigma$ if $(p,b,right) = \delta(q,a)$ with $a,b\in \Sigma$, $q\in Q\setminus Q_f$, $p\in Q$; %
\item $(\$ q'a'\$ , @\#''p''b'' @) $, if $(p,b,right) = \delta(q,a)$ with $a,b\in \Sigma$, $q\in Q\setminus Q_f$, $p\in Q$; %
(for the case when there was a right move from the last used tape cell);
\item $(\$ c'q'a' ,b''c''p'' @) $, for all $c\in\Sigma$ if $(p,b,left) = \delta(q,a)$ with $a,b\in \Sigma$, $b\ne \hash$, $q\in Q\setminus Q_f$, $p\in Q$; %
\item $(\$ c'q'a' ,c''p''@ ) $, for all $c\in\Sigma$ if $(p,\#,left) = \delta(q,a)$ with $a\in \Sigma$,  $q\in Q\setminus Q_f$, $p\in Q$; %
\item $(\$ q'a' ,b''\#''p'' @) $,  if $(p,b,left) = \delta(q,a)$ with $a,b\in \Sigma$, $b\ne \hash$, $q\in Q\setminus Q_f$, $p\in Q$; %
\item $(\$ q'a' ,\#''p'' @) $,  if $(p,\#,left) = \delta(q,a)$ with $a\in \Sigma$,  $q\in Q\setminus Q_f$, $p\in Q$; %
\item $(c' q'a', b''p'' c'') $, for all $c\in \Sigma$, if $(p,b,stay) = \delta(q,a)$ with $a,b\in \Sigma$, $q\in Q\setminus Q_f$, $p\in Q$; 
\item $(c' q'a'd' , d''p''b'' c'') $, for all $c, d\in \Sigma$, if $(p,b,right) = \delta(q,a)$ with $a,b\in \Sigma$, $q\in Q\setminus Q_f$, $p\in Q$; %
\item $(c' q'a'\$ , @\#''p''b'' c'') $, for all $c\in \Sigma$, if $(p,b,right) = \delta(q,a)$ with $a,b\in \Sigma$, $q\in Q\setminus Q_f$, $p\in Q$; %
(for a right move from the last used tape cell);
\item $(d' c'q'a' ,b''c''p'' d'') $, for all $c,d\in \Sigma$,  if $(p,b,left) = \delta(q,a)$ with $a,b\in \Sigma$, $b\ne \hash$, $q\in Q\setminus Q_f$, $p\in Q$; %
\item $(d' c'q'a' ,c''p''d'' ) $, for all $c, d\in\Sigma$ if $(p,\#,left) = \delta(q,a)$ with $a\in \Sigma$,  $q\in Q\setminus Q_f$, $p\in Q$; %
\item $(d' q'a' ,b''\#''p'' d'') $, for all $d\in \Sigma$,  if $(p,b,left) = \delta(q,a)$ with $a,b\in \Sigma$, $b\ne \hash$, $q\in Q\setminus Q_f$, $p\in Q$; %
\item $(d' q'a' ,\#''p'' d'') $, for all $d\in \Sigma$,  if $(p,\#,left) = \delta(q,a)$ with $a\in \Sigma$,  $q\in Q\setminus Q_f$, $p\in Q$. %
\end{itemize}

In this way,  in the upper strand a sequence $\$ c_0\$ c'_1\$ \dots\$ c'_{m-1} $ can be assembled while in the lower strand the string $ (c''_m)^R @ \dots @ (c''_1)^R$. No common characters are used, thus the two strands cannot overlap yet.

To construct the first appearance of an accepting configuration on both strands, let $T$ contain %
\begin{itemize}
\item $(\$ q'a'\$ c''',\$ \$ @\#''p''b'' @)$ for all $c\in \Sigma$, if $(p,b,right) = \delta(q,a)$ with $a,b\in \Sigma$, $q\in Q\setminus Q_f$, $p\in Q_f$; %
\item $(c'q'a'\$ d''',\$ \$ @\#''p''b''c'') $ for all $c, d\in \Sigma$, if $(p,b,right) = \delta(q,a)$ with $a,b\in \Sigma$, $q\in Q\setminus Q_f$, $p\in Q_f$; %
\item $(a'''b''',\$)$ for all $a,b\in \Sigma$; %
\item $(a'''p''',\$)$ for all $a\in \Sigma$, $p\in Q_f$; and
\item $(p''' \#''' \$ \$ @, \$)$  for each $p\in Q_f$.
\end{itemize}

This could be the first step when the two strands overlap. From here, we need assembly units to complete both strands, and check if the corresponding configurations coded already in the two strands are the same.

\begin{itemize}
\item $(@p'',p'''\$ )$ for each $p\in Q_f$ with its double and triple primed versions;
\item $(a''b'',b'''a'''\$ )$ for each $a,b\in \Sigma_1$ with their double and triple primed versions;
\item $(a''@,\$ a'''\$ )$ for each $a\in \Sigma_1$ with its double and triple primed version;
\item $(@a'',a'\$ )$ for each $a \in \Sigma_1$ with its primed and double primed versions;
\item $(a''b'',b'a'\$ )$ for each $a,b\in \Sigma_1$ with their primed and double primed versions;
\item $(a''@,\$ a'''\$ )$ for each $a\in \Sigma_1$ with its double and triple primed version.
\end{itemize}

The full upper strand can be made with these units, checking the identity  of the configurations written in double primed and triple primed or primed way, respectively. %

As the last part, $T$ must also have the following units:

\begin{itemize}
\item $(q_0'',\$ a')$ for each $a'\in \Sigma_2$;
\item $(q_0'', a \$ )$ for each $a\in V$;
\item $(q_0'', a b )$ for each $a,b \in V$.
\end{itemize}

Finally, $E$ contains the ending assembly units
\begin{itemize}
\item $(q_0'', \$q_0 a )$ for each $a\in V$, with the initial state $q_0$.
\end{itemize}

Now, one can see based on the overlapping symbols that the obtained \5 SAS generates only double strings of the following regular form:
$$ \$q_0 \Sigma^+ \$ (\Sigma^+_2 \$)^* \Sigma^+_4 \$ \$ (@ \Sigma^+_3 )^+ ,$$
and thus, it `simulates' exactly the accepting computations of the given Turing machine.

Since it is undecidable whether a TM has any accepting computations (Theorem \ref{thm:T}), %
 the emptiness problem of \5 SAS is also undecidable.
\end{proof}

Based on the proofs we had so far, we can also establish a new result characterizing RE, the class of recursively enumerable languages.

\begin{theorem}
Every recursively enumerable language $L$ can be obtained by an erasing morphism from a language of $\3$.
\end{theorem}

\section{The unary case}

So far, we have shown some undecidability results, on the other hand, the unary case is different: %

\begin{proposition}\label{prop:unaryPCP}
For the unary alphabet the reverse PCP is decidable.
\end{proposition}

\begin{proof} Let us consider the unary reverse PCP with a set of dominoes $D =\{ (\alpha_1, \beta_1), (\alpha_2, \beta_2),\dots, (\alpha_n, \beta_n) \}$ over the alphabet $\{1\}$. There are the following cases:

\begin{enumerate}
\item Let us assume that there is a domino  $d_i=(1^n,1^n)$, i.e., $\alpha_i = \beta_i$.
Then, taking this domino once (or any positive integer times) results in a solution of the problem. Thus, the problem is solvable. \\ \\ Further, let us assume that no domino with this property exist in $D$.
\item In the case when, for all domino $d_i \in D$, $| \alpha_i | > |\beta_i|$, i.e., the upper string of each domino is longer than its lower string, 
 it is easy to see that for any positive number of applications of any dominoes (or their combinations) the result has a longer upper string than a lower string, and thus no solution exists.
\item In the opposite case, i.e., if for each domino $d_i$,  $| \alpha_i | < |\beta_i|$, a similar argument by interchanging the roles of the upper and lower strings show that there is no solution.
\item In the remaining case, there are both types of dominoes, i.e., there is a domino $d_i \in D$ such that $| \alpha_i | > |\beta_i|$ and also there is a domino $d_\ell \in D$ such that $| \alpha_\ell | < |\beta_\ell|$.
Let us use the following notation: let  $n = | \alpha_i | $, $m = |\beta_i|$,  $k = | \alpha_\ell | $ and $j = |\beta_\ell|$. Then
by applying domino $d_i$ for $(j-k)$ times and domino $d_\ell$ for $(n-m)$ times will lead us to a solution. The problem is solvable.
\end{enumerate}

As all possible cases are covered above, the decidability of the problem is shown.
 \end{proof}

Moreover, in fact, in the above proof, we have characterized the solvable and the unsolvable unary reverse Post Corresponding Problems.

\begin{proposition}
\label{prop:unaryDEC}
For unary \5 SAS, the emptiness problem is decidable.
\end{proposition}

\begin{proof} Let us consider a unary \5 SAS $S=(\{1\}, A, T, E)$.

As both $A$ and $E$ are finite, there is a finite amount of pairs of their elements, and it is easy to check if for any pair of an axiom and an ending assembly unit results in a generated word, i.e., the length of the upper and lower strings are the same when applying them (without any additional assembling units). As in this case, the decidability has already been proven, in the rest of the proof we assume that $S$ has no such pair of axiom and ending assembly unit.

Then, in the negative case, when no such pair produce directly a derived word, we can obtain a finite set of integers, where for each pair, the difference of the lengths of its upper and lower strings are measured: a positive integer shows how much the upper string is longer than the lower string and (the absolute value of) a negative integer shows how much the lower string is longer than the upper string for the given pair. Let this set be $B= \{-a_m, -a_{m-1}, \dots, -a_1 , b_1, \dots, b_n \}$ where the values $-a_i$ are the negative and the values $b_i$ are the positive values (for some nonnegative integers $m,n$).

Now let us consider the pairs in $T$ and define the set $C= \{  -e_j, 
\dots, -e_1 , f_1, \dots,$ $ f_\ell  \}$ similarly (for nonnegative integers $j$ and $\ell$) by the negative and positive values of the length differences of the upper and lower strings of the elements of $T$. Observe that we do not care about assembly units in $T$ with equal length upper and lower strings.
In case, all elements of $T$ have this special property (i.e. each element contains identical strings), but there is no word that can be assembled by only using an axiom and ending unit pair, the given \5 SAS generates the empty language.

Now, we consider the remaining case: when both $B$ and $C$ are nonempty sets.

Clearly, if both $B$ and $C$ contains either only positive numbers or only negative numbers, the given \5 SAS $S$ generates the empty language.

Now, let us consider the case when $C$ contains only positive numbers, but $B$ has also negative value(s).
In this case $ j=0, \ell >0$ and $ m>0 $.

In this case there are finitely many possibilities to check:
first we choose a negative element $a_i$ of $B$ and, then we should check if there is a finite multiset $C'$ from the (positive) values of $C$ such that their sum  is 
 $|a_i|$. If we can find such $a_i$ and multiset $C'$ with the above property, then the language generated by $S$ is nonempty, otherwise it is empty.
(One needs to check only finitely many possibilities, as only multisets containing at most $|a_i|$ elements should be checked for each $a_i$.)

The case when $C$ contains only negative numbers, but $B$ has also positive value(s) is similar.
In this case $ j>0, \ell =0$ and $ n>0 $.

In the remaining case, $C$ has both positive and negative numbers. The question is if there is a multiset of them such that their sum is the opposite to an element of $B$.

If the greatest common divisor of the positive elements of $C$ is $1$, and also the greatest common divisor of the absolute values of the negative elements is $1$, then (by a theorem connected to Frobenius) 
there is a positive number $k$ such that

\begin{itemize}
\item for each integer $ i_1$ larger than $k$, there is a multiset of the positive elements of $C$ such that their sum is $i_1$, and
\item for each integer $ i_2$ larger than $k$, there is a multiset of the negative elements of $C$ such that the sum of their absolute value is $i_2$.
\end{itemize}

Based on that we can pick an element of $B$ (with its axiom and ending units) and multisets of elements of $T$ represented by the positive and negative elements of $C$ such that, they provide a generated word and the language is not empty.

Similar arguments work already if at least one of the greatest common divisors of the positive elements of $C$ and of the absolute values of the negative elements of $C$ is $1$, since then, there is a positive $k$ such that (depending on the case) each number that is greater than $k$ can be obtained as a sum of the elements of a multiset of positive elements of $C$, or each number that is less than $-k$ can be obtained as a sum of the elements of a multiset of negative elements of $C$, and thus clearly, it is enough to put enough many elements of the other multiset and choosing an element of $B$ such that the assigned elements would lead to a solution.

Even if the greatest common divisor of these two greatest common divisors 
 is $1$, we can manage based on the following fact. %
Let the greatest common divisor of the positive elements of $C$ be $f >1$ and the   greatest common divisor of the absolute values of the negative elements of $C$ be $e > 1$.
In fact, one can reduce the problem to the well-known Diophantine equations  $xf - ye = \pm 1$. Each of these equations has infinitely many integer solution pairs of $x$ and $y$, and also such solutions were  both $x$ and $y$ are high enough. When $B$ does not contain any of $+1$ and $-1$, but $j >0$ or $-j<0$, we may need to take $|j|$ times the elements of the multisets of the elements of $C$ representing a given solution.

The last case is when the   greatest common divisor of the positive elements of $C$ is $f$ and the   greatest common divisor of the absolute values of the negative elements of $C$ is $e$ such that $e$ and $f$ are both multiplier of the greatest common divisor $z>1$ of $e$ and $f$.

In this case, if set $B$ contains a multiplier (negative or positive) of $z$, the generated language is not empty, otherwise it is empty. Since by the elements of $C$ only multipliers of $z$ can be obtained. However, the ideas of the previous cases work instead of using any large enough integers by using any large enough multipliers of $z$.

The decision made in each case, thus the proof is finished.
 \end{proof}

\section{Conclusions and Future Work}

\5 SAS is a new model for string assembly \cite{53SAS-hier}.
A new variant of the PCP is defined and used as an analogue to show undecidability result also for \5 SAS. Over the unary alphabet, these questions are decidable. The border of decidability is an open question. We plan to continue to characterize the \5 SAS languages and their properties and also to investigate other variants of SAS with the \5 feature based on \cite{SAS2}. %
Another open question is whether the regular filter can be removed from the undecidability proof of the new PCP variants.

 \bibliographystyle{eptcs} \bibliography{References53SAS}

\end{document}